\documentclass[12pt]{article}
\usepackage{latexsym}
\usepackage[dvips]{color}
\usepackage{amssymb,dsfont,amsfonts}
\usepackage{graphicx}
\usepackage{epsfig}
\usepackage{amsmath}
\usepackage{mathrsfs}
\usepackage{natbib}
\usepackage{hyperref}

\renewcommand{\cite}[1]{\citet{#1}}

\newtheorem{theorem}{Theorem}[section]
\newtheorem{proposition}[theorem]{Proposition}
\newtheorem{corollary}[theorem]{Corollary}
\newtheorem{lemma}[theorem]{Lemma}
\newtheorem{definition}[theorem]{Definition}
\newtheorem{remark}[theorem]{Remark}
\newtheorem{example}[theorem]{Example}

\newcommand{\esssup}{\operatorname*{\mathrm{ess\,sup}}}

\newcommand{\EE}{\mathbb{E}}

\newcommand{\E}{\EE}

\newcommand{\qed}{\mbox{ }~\hfill~$\Box$ \vspace{1ex} }

\newenvironment{proof}{\noindent{\sc Proof: }}{ \qed }

\newcommand{\rmII}{\text{\it I\kern-.08em I\,}}
\newcommand{\rmIII}{\text{\it I\kern-.08em I\kern-.08em I\,}}
\newcommand{\rmIV}{\text{\it I\kern-.08em V\,}}

\renewcommand{\epsilon}{\varepsilon}

\begin{document}


\title{Financial Markets with Volatility Uncertainty}
\author{J\"org Vorbrink\thanks{I thank Frank Riedel for
valuable advice and Yongsheng Song for fruitful discussions.
Financial support through the German Research Foundation (DFG) and
the International Graduate College ``Stochastics and Real World
Models" Beijing--Bielefeld is gratefully acknowledged.} \\
Institute of Mathematical Economics -- Bielefeld University\\
jvorbrin@math.uni-bielefeld.de }
\date { 
December 15, 2010} \maketitle

\begin{abstract}\vspace{-0.2cm}
We investigate financial markets under model risk
caused by uncertain volatilities. 
For this purpose we consider a financial market that features
volatility
uncertainty. 
To have a
mathematical consistent framework we use the notion of
\emph{G--expectation} and its corresponding \emph{G--Brownian
motion} recently introduced by \cite{peng0}.
Our financial market consists of a riskless asset and a risky
stock with price process modeled by a geometric G--Brownian
motion.
We adapt the notion of arbitrage to this more complex situation
and consider stock price dynamics which exclude arbitrage
opportunities. Due to volatility uncertainty the market is not
complete any more. We establish the interval of no--arbitrage
prices for general European contingent claims and deduce explicit
results in a Markovian setting. 
\noindent
\end{abstract}
\medskip
{\footnotesize{ \it Key words and phrases:} Pricing of contingent
claims, incomplete markets, volatility uncertainty, G--Brownian motion stochastic calculus \\
{\it \hspace*{0.6cm} MSC 2000 numbers: 91G20}\\
{\it \hspace*{0.6cm} JEL classification codes: G13, D81, C61}}
\newpage


\section{Introduction}
Many choice situations exhibit ambiguity. At least since the
Ellsberg Paradox the occurrence of ambiguity aversion and its
effect on making economic decisions are well established. One
possible way to model decisions under ambiguity is to use multiple
priors. Instead of analyzing a problem in a single prior model as
in the classical subjective expected utility approach one focuses
on a multiple prior model to describe the agent's uncertainty
about the right probability distribution.\par These models have
gained much attention in recent studies. The decision theoretical
setting of multiple priors was introduced by \cite{gs} and
extended to a dynamic model by \cite{es2}.
\cite{maccheroni} generalize this model to so--called variational
preferences. Multiple priors appear naturally in monetary risk
measures as introduced by \cite{artzner} and its dynamic
extensions, see \cite{fs} for an overview.
\par
Most literature essentially concentrates on the modeling of
multiple priors with respect to some reference
measure.\footnote{See \cite{ozaki} or \cite{riedel}, for example.}
The standing assumption is that all priors are at least (locally)
absolutely continuous with regard to a given reference measure.
This is often a technical assumption in order to simplify
mathematics. However, it significantly affects the informative
value of the multiple prior model.\footnote{The reference measure
plays the role of fixing the sets of measure zero. This means that
the decision maker has perfect knowledge about sure events which
is obviously not always a reasonable property from an economic
point of view. In particular, with a filtration satisfying the
``usual conditions" which is mostly arranged one excludes economic
interesting models since the decision maker consequently knows
already at time zero what can happen and what not. This of course
does not reflect reality well as for instance the recent incidents
about the Greek government bonds illustrate.} In diffusion models,
by Girsanov's theorem these multiple prior models only lead to
uncertainty in the mean of the considered stochastic process, see
\cite{ec} or \cite{riedel2} for instance. Thus in Finance, these
multiple prior models just lead to drift uncertainty for the stock
price.
\par Obviously, one may imagine another source of uncertainty
that involves the risk described by the standard deviation of a
random variable. Especially in Finance this is of great relevance.
For instance, the price of an option written on a risky stock
heavily depends on the underlying volatility. Also the value of a
portfolio consisting of risky positions is strongly connected with
the volatility levels of the corresponding assets. One major
problem in practice is to forecast the prospective volatility
process in the market. In this sense it appears quite natural to
permit volatility uncertainty.\footnote{Volatility is very
sensitive with respect to changing market data which makes its
predictability difficult. It also reflects the market's sentiment.
Currently high implicit volatility levels suggest nervous markets
whereas low levels rather feature bullish mood. By taking many
models into account one may protect oneself against surprising
events due to misspecification.}
 In the sense of
risk measuring it is desirable to seize this risk.
\par Our paper may form the basis for working with volatility
uncertainty in Finance. There are many problems like hedging of
contingent claims under constraints, cf. \cite{kk1}, or portfolio
optimization, cf. \cite{merton}, that should also be formulated
and treated in the presence of model uncertainty. For this one
will also need an economic reasonable notion of arbitrage as
developed in this paper.\par \cite{fernholz} consider the question
of outperforming the market when the assumption of no--arbitrage
is not imposed and incorporate volatility uncertainty. Our purpose
is to model volatility uncertainty on financial markets under the
assumption of no--arbitrage. We set a framework for modeling this
particular uncertainty and treat the pricing and hedging of
European contingent claims. The setting is closely related to
model risk which clearly matters in the sense of risk management.
As we shall see allowing for uncertain volatilities will lead to
incomplete markets, therefore effects the pricing and hedging of
claims, and involves model risk. Our solution approach also
provides a method to measure this risk. We consider European
claims written on a risky stock $S$ which additionally features
volatility uncertainty. Roughly, $S$ is modeled by the family of
processes
\begin{align*}
dS^\sigma_t=r S^\sigma_t dt+\sigma_t S^\sigma_t dB_t
\end{align*}
where $B=(B_t)$ is a classical Brownian motion and $\sigma_t$
attains various values in $[\underline{\sigma},\overline{\sigma}]$
for all $t$. In this setting we aim to solve
\begin{align}\label{H}
\sup_{P\in\mathcal{P}} E^P(H_T\gamma_T^{-1})\quad\text{ and }
\quad\sup_{P\in\mathcal{P}} E^P(-H_T\gamma_T^{-1})
\end{align}
where $H_T$ denotes the payoff of a contingent claim at maturity
$T,\,\gamma_T^{-1}$ a discounting, and $\mathcal{P}$ presents a
set of various probability measures describing the model
uncertainty.\par
It is by no means clear whether the expressions
above are well--posed and how to choose $\mathcal{P}$ in this
case. As seen in \cite{denis} modeling uncertain volatilities
leads to a set of priors $\mathcal{P}$ which consists of mutually
singular probability measures.\footnote{To understand how things
are involved see Example \ref{exneu} in Section \ref{section3}.}
So when dealing with model uncertainty we need a consistent
mathematical framework enabling us to work with processes under
various measures at the same time. We utilize the framework of
sublinear expectation and G--normal distribution introduced by
\cite{peng0} in order to model and control model risk.
\par We
consider a Black--Scholes like market with uncertain volatilities,
i.e., the stock price $S$ is modeled as a geometric G--Brownian
motion
\begin{align}\label{s}
dS_t=rS_tdt +S_tdB_t, \quad S_0=x_0,
\end{align}
where the canonical process $B=(B_t)$ is a G--Brownian motion with
respect to a sublinear expectation $E_G$. $E_G$ is called
G--expectation. It also represents a particular coherent risk
measure that enables to quantify the model risk induced by
volatility uncertainty. For the construction see \cite{peng0} or
\cite{peng1}.
\par G--Brownian motion forms a very rich and
interesting new structure which generalizes the classical
diffusion model. It replaces classical Brownian motion to account
for model risk in the volatility component. Each $B_t$ is
G--normal distributed which resembles the classical normal
distribution. The function G characterizes the degree of
uncertainty.\par Throughout the paper we consider the case where G
denotes the function $y\mapsto
G(y)=\frac{1}{2}\overline{\sigma}^2y^+
-\frac{1}{2}\underline{\sigma}^2y^-$ with volatility bounds
$0<\underline{\sigma}<\overline{\sigma}$. Each $B_t$ has mean zero
but an uncertain variance varying between the bounds
$\underline{\sigma}^2t$ and $\overline{\sigma}^2t$. So when $B_t$
is evaluated by $E_G$ we have $E_G(B_t)=E_G(-B_t)=0$ and
$E_G(B_t^2)=\overline{\sigma}^2t\neq
-E_G(-B_t^2)=\underline{\sigma}^2t$. Consequently, the stock's
volatility is uncertain and incorporated in the process $(B_t)$.
The quadratic variation process is no longer deterministic. All
uncertainty of $B$ is concentrated in $\langle B\rangle$. It is
absolutely continuous w.r.t. Lebesgue measure and its density
satisfies $\underline{\sigma}^2\leq \frac{d\langle
B\rangle_t}{dt}\leq \overline{\sigma}^2$. \par The related
stochastic calculus, especially It\^o's integral, can also be
established with respect to G--Brownian motion, cf. \cite{peng1}.
Notions like martingales are replaced by G--martingales with the
same meaning as one would expect from classical probability
theory, cf. Definition \ref{G--martingale}.\par Even though from
the first view of Equation \eqref{s} it is hidden we are also in a
multiple prior setting. \cite{pengdenis} showed that the
G--framework developed in \cite{peng0} corresponds to the
framework of quasi--sure analysis.\footnote{See \cite{denis}, for
example.} They established that the sublinear expectation $E_G$
can be represented as an upper expectation of classical
expectations, i.e., there exists a set of probability measures
$\mathcal{P}$ such that $E_G[X]=\sup_{P\in\mathcal{P}}E^P[X]$.
\par It should be mentioned that the stipulated dynamics for
the stock price in \eqref{s} imply that the discounted stock price
is a symmetric G--martingale. The word ``symmetric" implies that
the corresponding negative process is also a G--martingale which
is not necessarily the case in the G--framework.\footnote{This
characteristic forms a fundamental difference to classical
probability theory. Its effect is also reflected in the results of
Section \ref{section4}.} For this stock price model we prove that
the induced financial market does not admit any arbitrage
opportunity.\footnote{In an accompanying paper we will give
further economic verification for the dynamics in \eqref{s}.} In
addition, this accords with classical Finance in which problems
like pricing and hedging of claims are solved with respect to a
risk neutral martingale measure such that the discounted price
process becomes a (local) martingale.\par The notion of
G--martingale plays an important role in our analysis (see page
\pageref{G--martingale} 
in the appendix for a deeper understanding).
\par
By \cite{soner1} a G--martingale $(M_t)$ solves the following
dynamic programming principle, see also Appendix \ref{A3}:
\begin{align*}
M_t= \esssup_{Q'\in\mathcal{P}(t,Q)}E^{Q'}(M_s|\mathcal{F}_t)\quad
Q-a.s.,~t\leq s,
\end{align*}
where $\mathcal{P}(t,Q):=\{Q'\in\mathcal{P}|Q'=Q \text{ on
}\mathcal{F}_t\}$. Thus, a G--martingale is a multiple prior
martingale as considered in \cite{riedel}. The dynamic programming
principle states that a G--martingale is a supermartingale for all
single priors and a martingale for an optimal prior. Using the
G--martingale representation theorem of \cite{song1}, cf. Theorem
\ref{martingale}, and Remark \ref{K=0} we obtain that a symmetric
G--martingale will be a martingale with respect to all single
priors involved. So, also from this point of view the imposed
dynamics on $S$ in \eqref{s} are economic reasonable.
\par In such an ambiguous financial market we analyze European contingent
claims concerning pricing and hedging. We extend the asset pricing
to markets with volatility uncertainty. The concept of
no--arbitrage will play a major role in our analysis. Due to the
additional source of risk induced by volatility uncertainty the
classical definition of arbitrage is no longer adequate. We
introduce a new arbitrage definition that fits to our multiple
prior model with mutually singular priors. We verify that our
financial market does not admit any arbitrage opportunity in this
modified sense.
\par Using the concept of no--arbitrage we establish detailed results providing a better
economic understanding of financial markets under volatility
uncertainty. We determine an interval of no--arbitrage prices for
general contingent claims. The bounds of this interval -- the
upper and lower arbitrage prices $h_{up}$ and $h_{low}$ -- are
obtained as the expected value of the claim's discounted payoff
with respect to G--expectation, see \eqref{H}. They specify the
lowest initial capital required to hedge a short position in the
claim, or long position, respectively.\footnote{The expression on
the left hand side in \eqref{H} may be interpreted as the ask
price the seller is willing to accept for selling the claim,
whereas the other represents the bid price the buyer is willing to
pay.} \par Since $E_G$ is a sublinear expectation we know that
$h_{low}\neq h_{up}$ in general which verifies the market's
incompleteness. All in all, any price being within the interval
$(h_{low},h_{up})$ is a reasonable initial price for a European
contingent claim in the sense that it does not admit
arbitrage.\par In a Markovian setting when the claim's payoff only
depends on the current stock price of its underlying we deduce
more structure about the upper and lower arbitrage prices via the
so--called Black--Scholes--Barenblatt PDE. We derive an explicit
representation for the corresponding super--hedging strategies and
consumption plans.\footnote{It is also called side--payments, cf.
\cite{fs}.} In the special situation when the payoff function
exhibits convexity (concavity) the upper arbitrage price solves
the classical Black--Scholes PDE with volatility equal to
$\overline{\sigma}$ ($\underline{\sigma}$), and, vice versa
concerning the lower arbitrage price. This corresponds to a
worst--case volatility analysis as in \cite{elkaroui5} and
\cite{avellaneda}.\par In the Markovian setting the same results
were also established in \cite{avellaneda}. However, the
mathematical framework in \cite{avellaneda} is rather intuitive
and presumably not sufficient for a general study. Our analysis
thus provides a rigorous foundation for the results in
\cite{avellaneda}. A recent preprint of \cite{soner2} also
contains relevant results concerning uncertain
volatility.\footnote{The authors use an approach different from
ours.}
\cite{denis} also obtained the identity
$h_{up}=E_G(H_T\gamma_T^{-1})$ for the case when
$H_T\in\mathcal{C}_b(\Omega)$ by using quasi--sure
analysis.\footnote{$\mathcal{C}_b(\Omega)$ denotes the space of
bounded continuous functions on the path space.} We formulate the
financial market in the presence of volatility uncertainty using
the G--framework of \cite{peng0}. Utilizing the set of multiple
priors $\mathcal{P}$ induced by $E_G$ we are able to investigate
the market by a more powerful approach -- a no--arbitrage concept
which also accounts for volatility uncertainty. Only due to the
properties of G--Brownian motion we are also able to obtain
explicit results like the PDE derivation in the Markovian setting.
In particular we realize the PDE of \cite{black} as a special case
when the volatility uncertainty is set equal to zero. Furthermore,
we consider the lower arbitrage price $h_{low}$ and obtain the
identities for claims with $H_T\in L_G^p(\Omega_T), p\geq
2.$\footnote{$L^p_G(\Omega_T)$ represents a specific space of
random variables for which the G--expectation can be defined.
$\mathcal{C}_b(\Omega)$ is contained in $L^p_G(\Omega_T)$, see
\cite{peng1} or Equation \eqref{pathwise} in Appendix \ref{A1}.}
Both bounds together form the basis for an economic reasonable
price for the claim.
\par
The paper is organized as follows. Section \ref{section3} introduces the
financial market we focus on and extends terminologies from
mathematical finance. Section \ref{section4} stresses on the
concept of no--arbitrage and the pricing of contingent claims in
the general case.
In Section \ref{section5} we restrict ourselves to the Markovian
setting and obtain similar results as in \cite{avellaneda}.
Section \ref{section6} concludes. Appendix \ref{section2}
introduces to the notions of sublinear expectation. It summarizes
the necessary definitions, constructions and associated results
from the original sources as far as we used it in the introduction
and preceding sections. So the appendix serves as a reference work
for the reader to obtain a deeper understanding of the
mathematics. Whenever we use a new item from the G--framework in
the previous sections we will give a brief and rough explanation
in order to facilitate a first reading without knowing the
framework being presented in the Appendix \ref{section2}.



\section{The market model}\label{section3}

We aim to analyze financial markets that feature volatility
uncertainty.
The following example (see also \cite{soner2}) illustrates some
issues that arise when we deal with uncertain volatilities.
\begin{example}\label{exneu}
Let $(B_t)$ be a Brownian motion with regard to some measure $P$
and consider the price process modeled as
$dS_t^\sigma=\sigma_tdB_t,~ S_0=x$, for various processes
$\sigma$. If $\hat{\sigma}\neq\tilde{\sigma}$ we have $\langle
S^{\hat{\sigma}}\rangle_\cdot\neq\langle
S^{\tilde{\sigma}}\rangle_\cdot$ $P$--almost surely which implies
that the distributions $P\circ(S^{\hat{\sigma}})^{-1}$ and
$P\circ(S^{\tilde{\sigma}})^{-1}$ are mutually singular.\par So,
given a family of stochastic processes $X^P, P\in\mathcal{P}$, we
need to construct a \emph{universal} process $X$ which is uniquely
defined with respect to all measures at the same time such that
$X=X^P ~P$--a.s. for all $P\in\mathcal{P}$. Also when defining a
stochastic integral $I_t^P:=\int_0^t\eta_sdB_s$ for all
$P\in\mathcal{P}$ simultaneously the same situation arises.
Clearly, we can define $I_t^P$ under each $P$ in the classical
sense. Since $I_t^P$ may depend on the respective underlying
measure $P$ we are free to redefine the integral outside the
support of $P$. Thus in order to make things work we need to find
a universal integral $I_t$ satisfying $I_t=I_t^P~P$--a.s. for all
measures $P\in\mathcal{P}$.
\end{example}

Let us now come to the introduction of the financial market. All
along the paper we consider a financial market $\mathcal{M}$
consisting of two assets evolving according to
\begin{align}\label{S1}
d\gamma_t=& r\gamma_tdt, ~\gamma_0=1, \nonumber\\
dS_t=& r S_tdt+S_tdB_t, ~S_0=x_0>0,
\end{align}
with a constant interest rate $r\geq 0$. $B=(B_t)$ denotes the
canonical process which is a G--Brownian motion under $E_G$ with
parameters $\overline{\sigma}>\underline{\sigma}>0$. For the exact
definition and construction of the pair $B$ and $E_G$ see Appendix
\ref{A1}. The assumption of strict positive volatility is well
accepted in Finance. It is of economic relevance, see also Remark
\ref{remark_sigma}. The asset $\gamma=(\gamma_t)$ represents a
riskless bond as usual. Since $B=(B_t)$ is a G--Brownian motion,
$S$ is modeled as a geometric G--Brownian motion similarly to the
original Back--Scholes model, cf. \cite{black}, where the stock
price is modeled by a classical geometric Brownian motion.\par As
a consequence, in this market the stock price evolution does not
only involve risk modeled by the noise part but also ambiguity
about the risk due to the unknown deviation of the process B from
its mean. In terms of Finance this ambiguity in the stock price is
called volatility uncertainty. If we choose
$\sigma=\underline{\sigma}=\overline{\sigma}$ we are in the
classical Black-Scholes model, see \cite{black} or any good
textbook in Finance.

\begin{remark}
Note that the discounted stock price process $(\gamma_t^{-1}S_t)$
is directly modeled as a symmetric G--martingale with regard to
the corresponding G--expectation $E_G$. It is a well known fact in
Finance that problems like pricing or hedging contingent claims
for instance are handled under a risk--neutral probability measure
which leads to the favored situation in which the discounted stock
price process is a (local) martingale, cf.
\cite{duffie}.\footnote{This should also be the case in our
ambiguous setting. By modeling the discounted stock price directly
as a symmetric G--martingale we do not have to change the
sublinear expectation from a subjective to the risk-neutral
sublinear expectation and avoid the technical difficulties
involved. In an accompanying paper we will give economic
verification for this dynamics.}
\end{remark}

The use of G--Brownian motion in order to model the financial
market initially leads to a formulation of $\mathcal{M}$ which is
not based on a classical probability space. The representation
theorem for G--expectation, see Theorem
\ref{theorem_representation}, establishes a link also to a
probabilistic framework. It provides us with a probability space
$(\Omega_T,\mathcal{F},P)$ and a set of multiple priors
$\mathcal{P}$ such that the following identity holds
\begin{align*}
E_G(X)=\sup_{P\in\mathcal{P}}E^P(X)
\end{align*}
where $X$ is any random variable for which the G--expectation can
be defined, for instance if $X:\Omega_T\rightarrow\mathbb{R}$ is
bounded and continuous. $\mathcal{F}=\mathcal{B}(\Omega_T)$
denotes the Borel $\sigma$--algebra on the path space
$\Omega_T=C_0([0,T],\mathbb{R})$. Besides, there exists a process
$W=(W_t)$ which is a classical Brownian motion w.r.t. $P$. We can
consider the filtration $(\mathcal{F}_t)$ generated by $W$, i.e.,
$\mathcal{F}_t:=\sigma\{W_s|0\leq s\leq t\}\vee \mathcal{N}$ where
$\mathcal{N}$ denotes the collection of $P$--null subsets. Then
the set of multiple priors $\mathcal{P}$ can be constructed as
follows.\par Let $\Theta:=[\underline{\sigma},\overline{\sigma}]$,
and $\mathcal{A}_{0,T}^\Theta$ be the collection of all
$\Theta$--valued $(\mathcal{F}_t)$--adapted processes on $[0,T]$.
For any $\theta\in\mathcal{A}_{0,T}^\Theta$ we define $
B_t^{0,\theta}:=\int_0^t \theta_s dW_s$ and $P^\theta$ as the law
of $B^{0,\theta}=\int_0^\cdot\theta_sdW_s$, i.e., $P^\theta=P\circ
(B^{0,\theta})^{-1}$. Then $\mathcal{P}$ is the closure of
$\{P^\theta|\theta\in\mathcal{A}_{0,T}^\Theta\}$ under
the topology of weak convergence.\\

\par
All along the paper we will consider the induced tuple
$(\Omega_T,\mathcal{F},(\mathcal{F}_t),W,P)$ together with the set
of priors $\mathcal{P}$ as given. Since $\mathcal{P}$ represents
$E_G$, it also represents the volatility uncertainty of the stock
price and therefore of $\mathcal{M}$. The G--framework utilized in
this paper enables the analysis of stochastic processes for all
priors of $\mathcal{P}$ simultaneously. The terminology of
``quasi--sure" turns out to be very useful:
\begin{center}
 A set $A\in\mathcal{F}$ is called \emph{polar} if $P(A)=0$ for all
$P\in\mathcal{P}$. We say a property holds ``\emph{quasi--surely}"
(q.s.) if it holds outside a polar set.
\end{center}
When not stated otherwise all equations are also to be understood
in the sense of ``quasi-sure". This means that a property holds
almost-surely for all conceivable scenarios.
\par
Next we repeat some useful definitions which are standard in
Finance but have to be adapted to this more complex situation. We
will need to use the following spaces $L^p_G(\Omega_T),
H^p_G(0,T),$ and also $M^p_G(0,T), p\geq 1$, which denote specific
spaces in the G--setting. The first one concerns random variables
for which the G--expectation is defined, see Equation
\eqref{pathwise} in Appendix \ref{A1}. The other two are
particular spaces of processes for which stochastic integrals with
respect to $B$ or $\langle B\rangle_\cdot$, respectively, can be
defined. They are the closure of collections of simple processes
similar to the case when the classical It\^o integral is
constructed, see also Appendix \ref{A2} at pages \pageref{form} --
\pageref{G--martingale}.
\par All along the paper we will presume a
finite time horizon denoted by $T>0$.
\begin{definition}\label{consumption}
A \emph{trading strategy} in the market $\mathcal{M}$ is an
$(\mathcal{F}_t)-$adapted vector process
$(\eta,\phi)=(\eta_t,\phi_t)$, $\phi$ a member of $H_G^1(0,T)$
such that $(\phi_t S_t)\in H_G^1(0,T)$.\par A \emph{cumulative
consumption process} $C=(C_t)$ is a nonnegative
$(\mathcal{F}_t)-$adapted process with values in
$L^1_G(\Omega_T)$, and with increasing, RCLL paths on $(0,T]$, and
$C_0=0,~C_T<\infty$ q.s.
\end{definition}

Note that the stock's price process $S$ defined by \eqref{S1} is
an element of $M^2_G(0,T)$ which coincides with $H^2_G(0,T)$, see
\cite{peng1}. We impose the so--called self--financing condition,
that is, consumption and trading in $\mathcal{M}$ satisfy
\begin{align}\label{s.f.}
V_t:=\eta_t\gamma_t +\phi_tS_t= \eta_0\gamma_0 +\phi_0S_0
+\int_0^t\eta_ud\gamma_u+\int_0^t\phi_udS_u-C_t~~ \forall t\leq
T~q.s.
\end{align}
where $V_t$ denotes the value of the trading strategy at time
$t$.\par Sometimes, it is more appropriate to consider instead of
a trading strategy a \emph{portfolio process} which presents the
proportions of wealth invested in the risky stock.

\begin{remark}\label{proportion}
A portfolio process $\pi$ represents proportions of a wealth $X$
which are invested in the stock within the considered time
interval whereas a trading strategy $(\eta,\phi)$ represents the
total numbers of the respective assets the agent holds. Clearly,
there is a one-to-one correspondence between a portfolio process
and a trading strategy as defined above. If we define
\begin{align*}
\phi_t:=&\frac{X_t\pi_t}{S_t},\quad\eta_t:=\frac{X_t(1-\pi_t)}{\gamma_t},\quad\forall
t\leq T,
\end{align*}
$(\eta,\phi)$ constitutes a trading strategy in the sense of
equation \eqref{s.f.} as long as $\pi$ constitutes a portfolio
process with corresponding wealth process $X$ as required in
Definition \ref{def5} below.
\end{remark}

\begin{definition}\label{pfoprocess}
A \emph{portfolio process} is an $(\mathcal{F}_t)-$adapted real
valued process $\pi=(\pi_t)$ with values in $L^1_G(\Omega_T)$.
\end{definition}

\begin{definition}\label{def5}
For a given initial capital $y$, a portfolio process $\pi$, and a
cumulative consumption process $C$, consider the wealth equation
\begin{align*}
dX_t=&~
X_t(1-\pi_t)\frac{d\gamma_t}{\gamma_t}+X_t\pi_t\frac{dS_t}{S_t}-dC_t\\
=&~ X_trdt+X_t \pi_tdB_t-dC_t
\end{align*} with initial wealth $X_0=y$.
Or equivalently,
\begin{align*}
\gamma_t^{-1}X_t=y-\int_0^t
\gamma_u^{-1}dC_u+\int_0^t\gamma_u^{-1}X_u\pi_udB_u,~~\forall
t\leq T.
\end{align*}
If this equation has a unique solution $X=(X_t):=X^{y,\pi,C}$ it
is called the \emph{wealth process} corresponding to the triple
$(y,\pi,C)$.
\end{definition}

In order to have the stochastic integral well defined, $\int_0^T
X_t^2\pi_t^2dt<\infty$ must hold quasi--surely and we need to
impose the requirement that $(\pi_t X_t)\in H^p_G(0,T),p\geq 1,$
or $\in M^p_G(0,T),p\geq 2$. We incorporate this into the next
definition which describes admissible portfolio processes.

\begin{definition}\label{admissible}
A portfolio/consumption process pair $(\pi,C)$ is called
admissible for an initial capital $y\in\mathbb{R}$ if
\begin{itemize}
\item [(i)] the pair obeys the conditions of Definitions
\ref{consumption} -- \ref{def5} \item [(ii)] $(\pi_t
X^{y,\pi,C}_t)\in H_G^1(0,T)$ \item[(iii)] the solution
$X^{y,\pi,C}_t$ satisfies
\begin{align*} X^{y,\pi,C}_t\geq -L, ~~\forall t\leq T,~ q.s.
\end{align*}
where $L$ is a nonnegative random variable in $L^2_G(\Omega_T)$.
\end{itemize}
We then write $(\pi,C)\in\mathcal{A}(y)$.
\end{definition}

In the above Definitions \ref{consumption} -- \ref{admissible} we
need to assure that the associated stochastic integrals are
well--defined. In particular condition (ii) of Definition
\ref{admissible} ensures that the mathematical framework does not
collapse by allowing for too many portfolio processes.
\par
The agent is uncertain about the true volatility, therefore, she
uses a portfolio strategy which can be performed independently of
the realized scenario at the market. Hence, she is able to analyze
the corresponding wealth processes with respect to all conceivable
market scenarios $P\in\mathcal{P}$ simultaneously.\par These
restrictions on the portfolio and consumption processes replace
the classical condition of predictable processes. Decisions at
some time $t$ must not utilize information which is revealed
subsequently. In our financial setting, the processes have to be
members of particular spaces within the G--framework. Based on the
construction of these spaces (by means of (viscosity) solutions of
PDEs, cf. Appendix \ref{section2}) the portfolio and consumption
processes require some kind of regularity, in particular see
identity \eqref{pathwise} in Appendix \ref{A1}. The economic
interpretation is that decisions should not react too abruptly and
sensitive to revealed information.


\section{Arbitrage and contingent claims}\label{section4}

As usual in financial markets we impose the concept of arbitrage.
Due to this more complex framework both of economic and
mathematical manner we need a slightly more involved definition of
arbitrage.

\begin{definition}[Arbitrage in $\mathcal{M}$]\label{arb0}
We say there is an arbitrage opportunity in $\mathcal{M}$ if there
exist an initial wealth $y\leq 0$, an admissible pair
$(\pi,C)\in\mathcal{A}(y)$ with $C\equiv 0$ such that at some time
$T>0$
\begin{align*}
X^{y,\pi,0}_T&\geq 0 ~~q.s.,\quad \text{and}\\
P\left(X^{y,\pi,0}_T> 0\right)&>0 ~~\text{for at least one
}P\in\mathcal{P}.
\end{align*}
\end{definition}

If such a strategy in the sense above existed one should pursue
this strategy since it would be riskless and in the lucky
situation that $P$ drove the market dynamics one would make a
profit with positive probability. It is important to note that in
the definition of arbitrage we have to require quasi--sure
dominance for the wealth at time $T$ in order to exclude the risk
in all possible scenarios that may occur. So there should not
exist a scenario under which there is positive probability that
the terminal wealth is less than zero.

\begin{remark}\label{discussion1}
The second condition in the definition of arbitrage is just the
negation of $X_T^{y,\pi,C}\leq 0$ q.s. Hence, combined with the
first condition it excludes that $X_T^{y,\pi,C}$ equals zero
quasi--surely.\par We identify $(y,\pi,C)$ as an arbitrage if
there exists profit with positive probability in at least one
scenario even though there does not exist profit with positive
probability in many others. Of course, one could also define
arbitrage by the requirement that the second condition has to hold
for all scenarios, i.e., there existed profit with positive
probability in all scenarios. We think that this kind of arbitrage
definition is not very reasonable from an economic point of view,
see Remark \ref{discussion2}.
\end{remark}

\begin{lemma}[No--arbitrage]\label{no-arb}
In the financial market $\mathcal{M}$ there does not exist any
arbitrage opportunity.
\end{lemma}
\begin{proof}
Assume there exists an arbitrage opportunity, i.e., there exists
some $y\leq 0$ and a pair $(\pi,C)\in\mathcal{A}(y)$ with $C\equiv
0$ such that $X_T^{y,\pi,0}\geq 0$ quasi--surely for some $T>0$.
Then we have $E_G(X_T^{y,\pi,0})\geq 0$. By definition of the
wealth process
\begin{align*}
0 \leq E_G\left(X^{y,\pi,0}_T\gamma_T^{-1}\right)\leq
y+E_G\left(\int_0^T \gamma_t^{-1}X^{y,\pi,0}_t\pi_tdB_t\right)=y
\end{align*}
since the G--expectation of an integral with respect to
G--Brownian motion is zero.\footnote{Here we used that
$(X^{y,\pi,0}_t\pi_t)\in H^1_G(0,T)$, cf. condition (ii) of
Definition \ref{admissible}.} Hence,
$E_G\left(X^{y,\pi,0}_T\gamma_T^{-1}\right)=0$ which again implies
$X_T^{y,\pi,0}\gamma_T^{-1}=0$ q.s. Thus, $(y,\pi,0)$ is not an
arbitrage.
\end{proof}\\

In the financial market $\mathcal{M}$ we want to consider European
contingent claims $H$ with payoff $H_T$ at maturity time $T$.
Here, $H_T$ represents a nonnegative, $\mathcal{F}_T$--adapted
random variable. All the time we impose the assumption $H_T\in
L^2_G(\Omega_T)$. The price of the claim at time $0$ will be
denoted by $H_0$. In order to find reasonable prices for $H$ we
use the concept of arbitrage. Similar to above we define an
arbitrage opportunity in the financial market $(\mathcal{M},H)$
consisting of the original market $\mathcal{M}$ and the contingent
claim $H$.

\begin{definition}[Arbitrage in $(\mathcal{M},H)$]\label{arbitrage}
There is an arbitrage opportunity in $(\mathcal{M},H)$ if there
exist an initial wealth $y\geq 0$ (respectively, $y\leq 0$), an
admissible pair $(\pi,C)\in\mathcal{A}(y)$ and a constant $a=-1$
(respectively, a=1), such that
\begin{align*}
y+a\cdot H_0\leq 0
\end{align*}
at time $0$, and
\begin{align*}
X^{y,\pi,C}_T&+a\cdot H_T\geq 0 ~~q.s.,\quad \text{and}\\
P\left(X^{y,\pi,C}_T+a\cdot H_T> 0\right)&>0 ~~\text{for at least
one }P\in\mathcal{P}
\end{align*}
at time $T$.
\end{definition}

The values $a=\pm 1$ in Definition \ref{arbitrage} indicate long
or short positions in the claim $H$, respectively. This definition
of arbitrage is standard in the literature, see \cite{karatzas2}.
For the same reasons as before we again require quasi--sure
dominance for the wealth at time $T$ and gain with positive
probability for only one possible scenario.\\  \par In the
following we show the existence of no--arbitrage prices for a
claim $H$ which exclude arbitrage opportunities. Compared to the
classical Black--Scholes model there are many no--arbitrage prices
for $H$ in general. We shall see that mostly hedging, or
replicating arguments, respectively, fail due to the additional
source of uncertainty induced by the G--normal distribution
causing the incompleteness of the financial market, see Remark
\ref{remarkK}. Thus in our ambiguous market $\mathcal{M}$ there
generally is either a self-financing portfolio strategy which
replicates the European claim nor a risk--free hedge for the claim
since the uncertainty represented by the occurring quadratic
variation term cannot be eliminated. Only for special claims $H$
when $\left(E_G[H_T|\mathcal{F}_t]\right)$ is a symmetric
G--martingale, cf. Remark \ref{remarkK}, we have $h_{up}=h_{low}$.
\par Clearly, there is only one single
G--Brownian motion which occurs in the financial market model.
However, due to the representation theorem for G--expectation
there are many probability measures involved in $\mathcal{M}$, cf.
Theorem \ref{theorem_representation}. Each measure reflects a
specific volatility rate for the stock price. Roughly speaking,
these measures induce the incompleteness since only one scenario
is being realized and only in this scenario the stock is being
traded. \par The functional $E_G$ is just a useful method to
control the dynamics by giving upper and lower bounds for European
contingent claim prices written on the stock, see Theorem
\ref{theorem_main}.\par
The following classes will matter in our
subsequent analysis.
\begin{definition}
Given a European contingent claim $H$ we define the lower hedging
class
\begin{align*}
\mathcal{L}:=\{y\geq 0|\exists~(\pi,C)\in\mathcal{A}(-y)
:~X^{-y,\pi,C}_T\geq -H_T~~q.s.\}
\end{align*}
and the upper hedging class
\begin{align*}
\mathcal{U}:=\{y\geq
0|\exists~(\pi,C)\in\mathcal{A}(y):~X^{y,\pi,C}_T\geq H_T~~q.s.\}.
\end{align*}
In addition, the lower arbitrage price is defined as
\begin{align*}
h_{low}:=\sup\{y|y\in\mathcal{L}\}
\end{align*}
and the upper arbitrage price as
\begin{align*}
h_{up}:=\inf\{y|y\in\mathcal{U}\}.
\end{align*}
\end{definition}

The main result of this section concerns the lower and upper
arbitrage price. It is possible to determine the prices
explicitly. We have
\begin{theorem}\label{theorem_main} Given the financial market
$(\mathcal{M},H)$. The following identities hold:
\begin{align*}
h_{up}=&\,E_G(H_T\gamma_T^{-1})\\
h_{low}=&-E_G(-H_T\gamma_T^{-1}).
\end{align*}
\end{theorem}

Before proving the theorem we establish some results about the
hedging classes. As proved in \cite{kk1} one can easily show that
$\mathcal{L}$ and $\mathcal{U}$ are connected intervals. More
precisely we have
\begin{lemma}\label{easy}
$y\in\mathcal{L}$ and $0\leq z\leq y$ implies $z\in\mathcal{L}$.
Analogously, $y\in\mathcal{U}$ and $z\geq y$ implies
$z\in\mathcal{U}$.
\end{lemma}

The proof uses the idea that one ``just consumes immediately the
difference of the two initial wealth". To include the case
$\mathcal{U}=\emptyset$ we define $\inf \emptyset=\infty$.\\
\par
For $\sigma\in [\underline{\sigma},\overline{\sigma}]$ let us
define the Black--Scholes price of a European contingent claim $H$
\begin{align*}
u_0^\sigma:= E^{P^\sigma}(H_T\gamma_T^{-1})
\end{align*}
where $P^\sigma\in\mathcal{P}$ denotes the measure under which $S$
has constant volatility level $\sigma$. As mentioned in Appendix
\ref{A1} it is defined by $P^\sigma:= P_0\circ (X^\sigma)^{-1}$
where $X^\sigma_t:=\int_0^t\sigma dW_u$. Clearly, due to the
dynamics of $S$, cf. Equation \ref{S1}, $P^\sigma$ is the usual
risk neutral probability measure in the Black--Scholes model with
fixed volatility rate $\sigma$.
\par
Similar as in the case with constraints, see \cite{kk1}, we can
prove the following three lemmata. For this let $H$ be a given
European contingent claim.
\begin{lemma}\label{lemma2}
For any $\sigma\in [\underline{\sigma},\overline{\sigma}]$ the
following inequality chain holds:
\begin{align*}
h_{low}\leq u_0^{\sigma}\leq h_{up}.
\end{align*}
\end{lemma}
\begin{proof}
Let $y\in\mathcal{U}$. By definition of $\mathcal{U}$ there exists
a pair $(\pi,C)\in\mathcal{A}(y)$ such that $X_T^{y,\pi,C}\geq
H_T$ q.s. Using the properties of G--expectation as stated in
Appendix \ref{A1}, in particular Prop. \ref{prop2} for the first
equality, we obtain for any $\sigma\in
[\underline{\sigma},\overline{\sigma}]$
\begin{align*}
y=&~E_G\left(y+\int_0^T\gamma_t^{-1}X_t^{y,\pi,C}\pi_tdB_t\right)\geq
E_G\left(y+\int_0^T\gamma_t^{-1}X_t^{y,\pi,C}\pi_tdB_t
-\int_0^T\gamma_t^{-1}dC_t\right)\\
=&~E_G\left(X_T^{y,\pi,C}\gamma_T^{-1}\right) \geq
E_G\left(H_T\gamma_T^{-1}\right)=
\sup_{P\in\mathcal{P}}E^P(H_T\gamma_T^{-1}) \geq u_0^\sigma.
\end{align*}
The first and second inequalities hold due to the monotonicity of
$E_G$, the second equality holds by the definition of the wealth
process and due to $y\in\mathcal{U}$, the third equality by the
representation theorem for $E_G$, cf. Theorem
\ref{theorem_representation}, and the last estimate holds because
of $P^\sigma\in\mathcal{P}$. Hence, $h_{up}\geq u_0^\sigma$.\par
Similarly, let $y\in\mathcal{L}$ and $(\pi,C)\in\mathcal{A}(-y)$
be the corresponding pair such that $X_T^{-y,\pi,C}\geq -H_T$ q.s.
By the same reasoning as above we obtain for any $\sigma\in
[\underline{\sigma},\overline{\sigma}]$
\begin{align*}
-y=&~E_G\left(-y+\int_0^T\gamma_t^{-1}X_t^{-y,\pi,C}\pi_tdB_t\right)\\
\geq &~E_G\left(-y+\int_0^T\gamma_t^{-1}X_t^{-y,\pi,C}\pi_tdB_t
-\int_0^T\gamma_t^{-1}dC_t\right)\\
=&~E_G\left(X_T^{-y,\pi,C}\gamma_T^{-1}\right) \geq
E_G\left(-H_T\gamma_T^{-1}\right)\geq
-E^{P^\sigma}\left(H_T\gamma_T^{-1}\right)= -u_0^\sigma
\end{align*}
which implies $y\leq u_0^\sigma$ and the Lemma follows.
\end{proof}
\begin{lemma}\label{lemma3}
For any price $H_0>h_{up}$ there exists an arbitrage opportunity.
Also for any price $H_0<h_{low}$ there exists an arbitrage
opportunity.
\end{lemma}
\begin{proof}
We only consider the first case since the argument is similar.
Assume $H_0>h_{up}$ and let $y\in (h_{up},H_0)$. By definition of
$h_{up}$ we deduce that $y\in\mathcal{U}$. Hence there exists a
pair $(\pi,C)\in\mathcal{A}(y)$ with
$$
X_T^{y,\pi,C}\geq H_T~~q.s.
$$
and $$y-H_0< 0.$$ This implies the existence of arbitrage in the
sense of Definition \ref{arbitrage}:\\ $\exists\, a>1$ with
$ay=H_0.$ Then $(\pi,aC)\in\mathcal{A}(ay)$ and
$X_T^{ay,\pi,aC}=aX_T^{a,\pi,C}$. Let $P\in\mathcal{P}$, w.l.o.g.
we may assume $P(H_T>0)>0$. Due to
\begin{align*}
1=P(X_T^{y,\pi,C}\geq H_T)\leq
P(aX_T^{y,\pi,C}>H_T)+P(X_T^{y,\pi,C}=0=H_T)
\end{align*}
we deduce $P(X_T^{ay,\pi,aC}>H_T)>0$. Hence, $(ay,\pi,aC)$
constitutes an arbitrage.
\end{proof}
\begin{lemma}\label{lemma4}
For any $H_0\notin \mathcal{L}\cup\mathcal{U}$ the financial
market $(\mathcal{M},H)$ is arbitrage free.
\end{lemma}
\begin{proof}
Assume $H_0\notin\mathcal{U},~H_0\notin\mathcal{L}$ and that there
exists an arbitrage opportunity in $(\mathcal{M},H)$. We suppose
that it satisfies Definition \ref{arbitrage} for $a=-1$. The case
$a=1$ works similarly.\\
By definition of arbitrage there exist $y\geq
0,~(\pi,C)\in\mathcal{A}(y)$ with
$$ y=X_0^{(y,\pi,C)}\leq H_0$$
and
$$X_T^{y,\pi,C}\geq H_T ~~q.s.$$
Hence, $y\in\mathcal{U}$, whence $H_0\in\mathcal{U}$ by Lemma
\ref{easy}. This contradicts our assumption.
\end{proof}\\

Now we pass to the proof of Theorem \ref{theorem_main}.\\

\begin{proof}
Let us begin with the first identity
$h_{up}=E_G(H_T\gamma_T^{-1}).$ As seen in the proof of Lemma
\ref{lemma2}, for any $y\in\mathcal{U}$ we have $y\geq
E_G\left(H_T\gamma_T^{-1}\right)$. Hence, $h_{up}=\inf
\{y|y\in\mathcal{U}\}\geq E_G\left(H_T\gamma_T^{-1}\right)$.
\par
To show the opposite inequality define the G-martingale $M$ by
\begin{align*}
M_t:=E_G\left(H_T\gamma_T^{-1}\right|\mathcal{F}_t)\quad\forall
t\leq T.
\end{align*}
By the martingale representation theorem (\cite{song1}), see
Theorem \ref{martingale}, there exist $z\in H_G^1(0,T)$ and a
continuous, increasing process $K=(K_t)$ with $K_T\in
L^1_G(\Omega_T)$ such that for any $t\leq T$
\begin{align*}
M_t=E_G(H_T\gamma_T^{-1}) +\int_0^tz_sdB_s-K_t ~~q.s.
\end{align*}
For any $t\leq T$ we set $y=E_G(H_T\gamma_T^{-1})\geq 0$,
$X_t\pi_t=z_t\gamma_t \in H_G^1(0,T)$, and
$C_t=\int_0^t\gamma_sdK_s\in L^1_G(\Omega_T)$. Then the induced
wealth process $X^{y,\pi,C}$ satisfies for any $t\leq T$
\begin{align*}
\gamma_t^{-1}X_t^{y,\pi,C}=y+\int_0^t
X^{y,\pi,C}_s\pi_s\gamma_s^{-1}dBs-\int_0^t\gamma_s^{-1}dC_s=M_t.
\end{align*}
$C$ obeys the conditions of a cumulative consumption process in
the sense of Definition \ref{consumption} due to the properties of
$K$. Because of $\gamma_t^{-1}X_t^{y,\pi,C}=M_t\geq 0~\forall
t\leq T$ the wealth process is bounded from below, whence
$(\pi,C)$ is admissible for
$y$.\\
As $X_T^{y,\pi,C}=\gamma_T M_T=H_T$ quasi--surely we have
$y=E_G(H_T\gamma_T^{-1})\in\mathcal{U}$. Due to the definition of
$\mathcal{U}$ we conclude $h_{up}\leq E_G(H_T\gamma_T^{-1})$.\\

The proof for the second identity is similar. Again, using the
proof of Lemma \ref{lemma2} we obtain $y\leq
-E_G\left(-H_T\gamma_T^{-1}\right)$ for any $y\in\mathcal{L}$ and
therefore $h_{low}\leq -E_G\left(-H_T\gamma_T^{-1}\right)$.

To obtain $h_{low}\geq -E_G(-H_T\gamma_T^{-1})$ we again define a
G-martingale $M$ by
\begin{align*}
M_t=E_G(-H_T\gamma_T^{-1}|\mathcal{F}_t)~~\forall t\leq T.
\end{align*}
The remaining part is almost a copy of above. Again by the
martingale representation theorem (\cite{song1}) there exist $z\in
H_G^1(0,T)$ and a continuous, increasing process $K=(K_t)$ with
$K_T\in L^1_G(\Omega_T)$ such that for any $t\leq T$
\begin{align*}
M_t=E_G(-H_T\gamma_T^{-1}) +\int_0^tz_sdB_s-K_t ~~q.s.
\end{align*}
As above, for any $t\leq T$ we set $-y=E_G(-H_T\gamma_T^{-1})\geq
0$, $X_t\pi_t=z_t\gamma_t \in H_G^1(0,T)$, and
$C_t=\int_0^t\gamma_sdK_s\in L^1_G(\Omega_T)$. Then the induced
wealth process $X^{-y,\pi,C}$ satisfies for all $t\leq T$
\begin{align*}
\gamma_t^{-1}X_t^{-y,\pi,C}=-y+\int_0^t
X^{-y,\pi,C}_s\pi_s\gamma_s^{-1}dBs-\int_0^t\gamma_s^{-1}dC_s=M_t.
\end{align*}
Again $C$ obeys the conditions of a cumulative consumption process
due to the properties of $K$. Furthermore, for any $t\leq T$
\begin{align*}
\gamma_t^{-1}X_t^{-y,\pi,C}=E_G\left(-H_T\gamma_T^{-1}|\mathcal{F}_t\right)\geq
E_G\left(-H_T|\mathcal{F}_t\right)
\end{align*}
which is bounded from below in the sense of item (iii) in
Definition \ref{admissible} since $-H_T\in L^2_G(\Omega_T)$. Hence
the wealth process is bounded from below. Consequently, $(\pi,C)$
is admissible for $-y$.\\
As $X_T^{-y,\pi,C}=\gamma_T M_T=-H_T$ quasi--surely we have
$y=-E_G(-H_T\gamma_T^{-1})\in\mathcal{L}$. Due to the definition
of $\mathcal{L}$ we conclude $h_{low}\geq -E_G(-H_T\gamma_T^{-1})$
which finishes the proof.
\end{proof}
\begin{remark}\label{remarkK}
By the last theorem we have $h_{low}\neq h_{up}$ in general since
$E_G$ is a sublinear expectation. This implies that the market is
incomplete meaning that not all claims can be hedged perfectly.
Thus in general, there are many no--arbitrage prices for $H$. We
always have $h_{low}\neq h_{up}$ as long as
$(E_G[H_T\gamma_T^{-1}|\mathcal{F}_t])$ is not a symmetric
G--martingale. In the other case, the process $K$ is identically
equal to zero, cf. Remark \ref{K=0}, implying that
$(E_G[H_T\gamma_T^{-1}|\mathcal{F}_t])$ is symmetric and $H_T$ can
be hedged perfectly due to Theorem \ref{martingale} and Remark
\ref{K=0}. As it is being showed in Section \ref{section5}, if $H$
for instance is the usual European call or put option this is only
the case if $\underline{\sigma}=\overline{\sigma}$ which again
implies that $E_G$ becomes the classical expectation.
\end{remark}
\begin{remark}
Again under the presumption of $h_{low}\neq h_{up}$ it is not
clear a priori whether a claim's price $H_0$ equal to $h_{up}$ or
$h_{low}$ induces an arbitrage opportunity or not. In the setting
of \cite{kk1} there may be situations where there is no arbitrage,
while in others there may be arbitrage. For instance, if
$H_0=h_{up}\in\mathcal{U}$ and $C_T>0$ a.s., then this consumption
can be viewed as kind of arbitrage opportunity (see \cite{kk1}).
The agent consumes along the way, and ends up with terminal wealth
$H_T$ almost surely.\par As seen in the proof of Theorem
\ref{theorem_main}, in our setting we always have
$h_{up}\in\mathcal{U}$ and $h_{low}\in\mathcal{L}$. We shall see
that due to our definition of arbitrage -- $
P\left(X_T^{y,\pi,C}-a H_T>0 \right)>0$ only has to hold for one
$P\in\mathcal{P}$ -- we have that a price $H_0=h_{up}$ or
$H_0=h_{low}$ induces arbitrage in $(\mathcal{M},H)$ in the sense
of Definition \ref{arbitrage}.
\end{remark}

\begin{corollary}
For any price $H_0\in (h_{low},h_{up})\neq \emptyset$ of a
European contingent claim at time zero there does not exist any
arbitrage opportunity
in $(\mathcal{M},H)$.\\
For any price $H_0\notin (h_{low},h_{up})\neq \emptyset$ there
does exist arbitrage in the market.
\end{corollary}
\begin{proof}
The first part directly follows from Lemma \ref{lemma4}. From
Lemma \ref{lemma3} we know that $H_0\notin [h_{low},h_{up}]$
implies the existence of an arbitrage opportunity. Therefore we
only have to show that $H_0=h_{up}$ and $H_0=h_{low}$ admits an
arbitrage opportunity.\\
We only treat the case $H_0=h_{up}$, the second case is analogue.
Comparing the proof of Theorem \ref{theorem_main}, for
$y=E_G(H_T\gamma_T^{-1})$ there exists a pair
$(\pi,C)\in\mathcal{A}(y)$ such that
\begin{align*}
\gamma_T^{-1}X_T^{y,\pi,C}=y+\int_0^T
X_s^{y,\pi,C}\pi_s\gamma_s^{-1}dB_s-\int_0^T\gamma_s^{-1}dC_s=H_T\gamma_T^{-1}~~q.s.
\end{align*}
We had $K_T=\int_0^T\gamma_s^{-1}dC_s$ where $K$ was an
increasing, continuous process with $E_G(-K_T)=0$. Hence we can
select $P\in\mathcal{P}$ such that $E_P(-K_T)<0$, see also Remark
\ref{remarkK}. Then the pair $(\pi,0)\in\mathcal{A}(y)$ satisfies
\begin{align*}
E^P\left(\gamma_T^{-1} X_T^{y,\pi,0}\right)>
E^P\left(\gamma_T^{-1}
X_T^{y,\pi,C}\right)=E^P\left(H_T\gamma_T^{-1}\right).
\end{align*}
Thus, $P\left(X_T^{y,\pi,0}>H_T\right)>0$ and we conclude that
$(\pi,0)\in \mathcal{A}(y)$ constitutes an arbitrage. So, possibly
the agent may consume along the way, and ends up with wealth $H_T$
quasi--surely.
\end{proof}
\begin{remark}\label{discussion2}
Note that the second statement of the corollary heavily depends on
the definition of arbitrage. Under the assumption of $h_{low}\neq
h_{up}$ it states that if $H_0$ is equal to one of the bounds
$h_{up}$ or $h_{low}$ there exists arbitrage in the sense of
Definition \ref{arbitrage}.\par Coming back to the discussion
about the definition of arbitrage started in Remark
\ref{discussion1} the proofs of the corollary and Theorem
\ref{theorem_main} also imply that if we required the last
condition in Definition \ref{arbitrage} to be true for all
scenarios $P\in\mathcal{P}$ then $H_0$ equal to one of the bounds
would not induce arbitrage in this new sense. Hence, $h_{up}$ and
$h_{low}$ would be reasonable prices for the claim.\par However,
there would exist profit with positive probability in many
scenarios. Only the scenarios $P\in\mathcal{P}$ that satisfy
$E^P(-K_T)=0$ would not provide profit with positive probability.
Thus, all $P\in\mathcal{P}$ not being maximizer of
$\sup_{P\in\mathcal{P}}E^P(-K_T)$ would induce arbitrage in the
classical sense when only one probability measure is involved.\par
From our point of view such a situation should be identified as
arbitrage which therefore supports our definition of arbitrage in
\ref{arb0} and \ref{arbitrage}. \par Additionally, even though our
arbitrage definition requires profit with positive probability for
only one scenario it is simultaneously satisfied for all
$P\in\mathcal{P}$ which are not maximizer of
$\sup_{P\in\mathcal{P}}E^P(-K_T)$.
\end{remark}
Based on the corollary we call $(h_{low},h_{up})\neq\emptyset$ the
arbitrage free interval. In the case where a more explicit
martingale representation theorem for
$(E_G[\gamma_T^{-1}H_T|\mathcal{F}_t])$ holds, see \cite{penghu},
we obtain a more explicit form for the consumption process $C$. In
particular in the Markovian setting where $H_T=\Phi(S_T)$ for some
Lipschitz function $\Phi:\mathbb{R}\rightarrow\mathbb{R}$ we can
give more structural details about the bounds $h_{up}$ and
$h_{low}$. We investigate this issue in the following section.


\section{The Markovian setting}\label{section5}

We consider the same financial market $\mathcal{M}$ as before and
restrict ourselves to European contingent claims $H$ which have
the form $H_T=\Phi(S_T)$ for some Lipschitz function
$\Phi:\mathbb{R}\rightarrow\mathbb{R}$.\par We will use a
nonlinear Feynman--Kac formula established in \cite{peng1}. For
this issue let us rewrite the dynamics of $S$ in \eqref{S1} as
\begin{align*}
dS^{t,x}_u= r S^{t,x}_udu+S^{t,x}_udB_u,~~ u\in [t,T],
~~S_t^{t,x}=x>0.
\end{align*}
Similar as the lower and upper arbitrage prices at time $0$ we
define the lower and upper arbitrage prices at time $t\in [0,T],
h_{low}^t(x)$ and $h_{up}^t(x)$. We use the variable $x$ just to
indicate that the stock at a considered time $t$ is at level $x$,
i.e., $S_t=x$.\par The following theorem is an extension of
Theorem \ref{theorem_main}. It establishes the connection of the
lower and upper arbitrage prices with solutions of partial
differential equations.
\begin{theorem}\label{theorem_summary}
Given a European contingent claim $H=\Phi(S_T)$ the upper
arbitrage price $h_{up}^t(x)$ is given by $u(t,x)$ where
$u:[0,T]\times\mathbb{R}_+\rightarrow\mathbb{R}$ is the unique
solution of the PDE
\begin{align}\label{BSB}
\partial_t u + rx\partial_x u+G\left(x^2\partial_{xx}u\right) = ru,
\quad u(T,x)=\Phi(x).
\end{align}
An explicit representation for the corresponding trading strategy
in the stock and the cumulative consumption process is given by
\begin{align*}
\phi_t=&~\partial_x u(t,S_t)\quad \forall t\in [0,T],\\
C_t=&~-\frac{1}{2}\int_0^t\partial_{xx}u(s,S_s)S_s^2 d\langle
B\rangle_s
+\int_0^tG\left(\partial_{xx}u(s,S_s)\right)S_s^2ds\quad\forall
t\in[0,T].
\end{align*}
Similarly, the lower arbitrage price $h_{low}^t(x)$ is given by
$-\underline{u}(t,x)$ where
$\underline{u}:[0,T]\times\mathbb{R}_+\rightarrow\mathbb{R}$ also
solves \eqref{BSB} but with terminal condition
$\underline{u}(T,x)=-\Phi(x)~\forall x\in\mathbb{R}_+$. Also, the
analog expressions hold true for the corresponding trading
strategy and cumulative consumption process.
\end{theorem}

The PDE in \eqref{BSB} is called Black--Scholes--Barenblatt
equation. It is also established in \cite{avellaneda}.\par Before
passing to the proof let us consider the BSDE
\begin{align*}
Y_s^{t,x}=E_G\left(\Phi(S_T^{t,x}) +\int_s^T
f(S_r^{t,x},Y_r^{t,x})dr | \mathcal{F}_s\right),\quad s\in [t,T],
\end{align*}
where $f:\mathbb{R}\times\mathbb{R}\rightarrow\mathbb{R}$ is a
given Lipschitz function. Since the BSDE has a unique solution,
see \cite{peng1}, we can define a function $u:
[0,T]\times\mathbb{R}_+\rightarrow\mathbb{R}$ by
$u(t,x):=Y_t^{t,x},~ (t,x)\in [0,T]\times \mathbb{R}_+$. Based on
a nonlinear version of the Feynman--Kac formula, see \cite{peng1},
the function $u$ is a viscosity solution of the following PDE
\begin{align}\label{pde}
\partial_t u + rx\partial_x u+G(x^2\partial_{xx}u) +f(x,u)=0,
\quad u(T,x)=\Phi(x).
\end{align}
Now we come to the proof of Theorem \ref{theorem_summary}.\\

\begin{proof}
It is enough just to treat the upper arbitrage price. For that
purpose define the function
\begin{align*}
\hat{u}(t,x):=E_G\left(\Phi(S_T^{t,x})\gamma_T^{-1}\right).
\end{align*}
By arguing as above $\hat{u}$ solves the PDE in \eqref{pde} for
$f\equiv 0$. Since the function $G$ is non--degenerate $\hat{u}$
even becomes a classical $C^{1,2}$--solution, see Remark
\ref{remark_sigma} or page 19 in \cite{peng1}. Therefore, together
with It\^o's formula (Theorem 5.4 in \cite{pengli}) we obtain
\begin{align*}
\hat{u}(t,S_t^{0,x})-&\hat{u}(0,x)\\ =&~\int_0^t \partial_t
\hat{u}(s,S_s^{0,x}) + rS_s^{0,x}\partial_x
\hat{u}(s,S_s^{0,x})ds\\ +&~ \int_0^t S_s\partial_x
\hat{u}(s,S_s^{0,x})dB_s+\int_0^t\frac{1}{2}S_s^2\partial_{xx}\hat{u}(s,S_s^{0,x})d\langle B\rangle_s\\
\overset{\eqref{pde}}{=}&~  \int_0^t S_s\partial_x
\hat{u}(s,S_s^{0,x})dB_s
\\+&~\underset{-K_t=-\int_0^t\gamma_s^{-1}dC_s}{\underbrace{\frac{1}{2}\int_0^t S_s^2\partial_{xx}\hat{u}(s,S_s^{0,x})d\langle B\rangle_s -\int_0^t
S_s^2G\left(\partial_{xx}\hat{u}(s,S_s^{0,x})\right)ds}}.
\end{align*}
Next consider the function
\begin{align*}
\tilde{u}(t,x):=\gamma_t\hat{u}(t,x),\quad\forall (t,x)\in
[0,T]\times\mathbb{R}_+.
\end{align*}
As in Theorem \ref{theorem_main} for $t=0$ we can deduce that
$\tilde{u}(t,x)=h_{up}^t(x)~\forall (t,x)\in
[0,T]\times\mathbb{R}_+$. In addition, one easily checks that
$\tilde{u}$ is a solution of the PDE in \eqref{BSB}. Also the
function $u$ defined by
\begin{align*}
u(t,x):=Y_t^{t,x}=E_G\left(\Phi(S_T^{t,x}) -\int_t^T rY_s^{t,x}ds
| \mathcal{F}_t\right)~\forall (t,x)\in [0,T]\times\mathbb{R}_+
\end{align*}
solves the PDE in \eqref{BSB} due to the nonlinear Feynman--Kac
formula since $f(x,y)=-ry$. By uniqueness of the solution in
\eqref{BSB}, see \cite{ishii} ($f$ is obviously bounded in $x$),
we conclude that $\tilde{u}=u$. Hence,
$u(t,x)=E_G\left(\Phi(S_T^{t,x})\gamma_{T-t}^{-1}\right)=h_t^{up}(x)~\forall
(t,x)\in [0,T]\times\mathbb{R}_+$ and it uniquely solves the
PDE.\par The explicit expressions for the trading strategy $\phi$
and the cumulative consumption process $C$ follow from the
calculations above for $\hat{u}$ and the identity
$\hat{u}(t,x)=u(t,x)\gamma_t^{-1}$.\par Comparing the proof of
Theorem \ref{theorem_main}, using its notations and Remark
\ref{proportion}, we obtain
$z_t=S_t^{0,x}\partial_x\hat{u}(t,S_t^{0,x})=\phi_tS_t^{0,x}\gamma_t^{-1}$.
Hence, $\phi_t=\partial_xu(t,S_t)~\forall t\in [0,T]$.\par
Similarly we derive
\begin{align*}
C_t=\int_0^t\gamma_sdK_s=-\frac{1}{2}\int_0^tS_s^2\partial_{xx}u(s,S_s)d\langle
B\rangle_s+\int_0^tS_s^2G\left(\partial_{xx}u(s,S_s)\right)ds.
\end{align*}
\end{proof}\\

Due to Theorem \ref{theorem_summary} the functions
$u(t,x)=h_{up}^t(x)$ and $\underline{u}(t,x)=-h_{low}^t(x)$ can be
characterized as the unique solutions of the
Black--Scholes--Barrenblatt equation. In the of case of $\Phi$
being a convex or concave function, respectively, the PDE in
\eqref{BSB} simplifies significantly. Due to the following result
it just becomes the classical Black--Scholes PDE in \eqref{BSB1}
for a certain constant volatility level.
\begin{lemma}\label{lemma5}
1. If $\Phi$ is convex $u(t,\cdot)$ is
convex for any $t\leq T$.\\
2. If $\Phi$ is concave $u(t,\cdot)$ is concave for any $t\leq T$.
Analogously, if $\Phi$ is convex $\underline{u}(t,\cdot)$ is
concave for any $t\leq T$. If $\Phi$ is concave
$\underline{u}(t,\cdot)$ is convex for any $t\leq T$.
\end{lemma}
\begin{proof}
Again we only need to consider the upper arbitrage price. It is
determined by the function
$u(t,x)=E_G\left(\Phi(S_T^{t,x})\gamma_{T-t}^{-1}\right)~\forall
(t,x)\in [0,T]\times\mathbb{R}_+.$\par
Firstly, let $\Phi$ be
convex, $t\in [0,T)$, and $x,y\in\mathbb{R}_+$. Then we have for
any $\alpha\in [0,1]$
\begin{align*}
u(t,\alpha x+&(1-\alpha)y)=E_G\left[\Phi\left(S_T^{t,\alpha
x+(1-\alpha)y}\right)e^{-r(T-t)}\right]\\
=&~E_G\left[\Phi\left(\left(\alpha
x+(1-\alpha)y\right)e^{r(T-t)-\frac{1}{2}<B>_{T-t}+B_{T-t}}\right)e^{-r(T-t)}\right]\\
\leq &~E_G\left[\alpha\Phi\left(
xe^{r(T-t)-\frac{1}{2}<B>_{T-t}+B_{T-t}}\right) +(1-\alpha)\Phi\left(ye^{r(T-t)-\frac{1}{2}<B>_{T-t}+B_{T-t}}\right)\right]e^{-r(T-t)}\\
\leq &~E_G\left[\alpha\Phi\left(
xe^{r(T-t)-\frac{1}{2}<B>_{T-t}+B_{T-t}}\right)e^{-r(T-t)}\right]\\& +E_G\left[(1-\alpha)\Phi\left(ye^{r(T-t)-\frac{1}{2}<B>_{T-t}+B_{T-t}}\right)e^{-r(T-t)}\right]\\
=&~ \alpha E_G\left[\Phi\left(S_T^{t,x}\right)
e^{-r(T-t)}\right]+(1-\alpha)E_G\left[\Phi\left(S_T^{t,y}\right)
e^{-r(T-t)}\right]\\ =&~\alpha u(t,x)+(1-\alpha)u(t,y)
\end{align*}
where we used the convexity of $\Phi$, the monotonicity of $E_G$
and in the second inequality the sublinearity of $E_G$. Thus,
$u(t,\cdot)$ is convex for all $t\in [0,T]$.\par Secondly, let
$\Phi$ be concave. Define for any $(t,x)\in
[0,T]\times\mathbb{R}_+$
\begin{align*}
v(t,x):=E^{P}\left[\Phi\left(\tilde{S}_T^{t,x}\right)e^{-r(T-t)}\right]
\end{align*}
where
\begin{align*}
d\tilde{S}^{t,x}_s= r
\tilde{S}^{t,x}_sds+\underline{\sigma}\tilde{S}^{t,x}_sdW_s,~~
s\in [t,T], ~~\tilde{S}_t^{t,x}=x.
\end{align*}
Remember that $W=(W_t)$ is a classical Brownian motion under $P$.
Then by the classical Feynman--Kac formula $v$ solves the
Black--Scholes PDE in \eqref{BSB1} with $\overline{\sigma}$
replaced by $\underline{\sigma}$.\par Since $E^{P}$ is linear it
is straightforward to show that $v(t,\cdot)$ is concave for any
$t\in [0,T]$. As a consequence, $v$ also solves \eqref{BSB}. By
uniqueness we conclude $v=u$. Hence, $u(t,\cdot)$ is concave for
any $t\in [0,T]$.
\end{proof}\\

As a consequence we have the following corollary.
\begin{corollary}
If $\Phi$ is convex $h_{up}^0(x)=
E^{P^{\overline{\sigma}}}\left(\Phi(S_T^{0,x})\gamma_T^{-1}\right)$
and
\begin{align*}
u(t,x):=
E_G\left(\Phi(S_T^{t,x})\gamma_{T-t}^{-1}\right)=E^{P^{\overline{\sigma}}}\left(\Phi(S_T^{t,x})\gamma_{T-t}^{-1}\right)
\end{align*}
solves the Black--Scholes PDE
\begin{align}\label{BSB1}
\partial_t u+rx\partial_x u+\frac{1}{2}\overline{\sigma}^2x^2\partial_{xx}u=ru,~~u(T,x)=\Phi(x).
\end{align}
If $\Phi$ is concave $h^0_{up}(x)=
E^{P^{\underline{\sigma}}}\left(\Phi(S_T^{0,x})\gamma^{-1}_T\right)$
and
\begin{align*}
u(t,x):=
E_G\left(\Phi(S_T^{t,x})\gamma_{T-t}^{-1}\right)=E^{P^{\underline{\sigma}}}\left(\Phi(S_T^{t,x})\gamma_{T-t}^{-1}\right)
\end{align*}
solves the PDE in \eqref{BSB1} with $\underline{\sigma}$ replacing
$\overline{\sigma}$.\\
Clearly, an analogue result holds for the lower arbitrage price
$h_{low}$, or terminal condition $\underline{u}(T,x)=-\Phi(x)$,
respectively.
\end{corollary}
\begin{proof}
The result directly follows from Theorem \ref{theorem_summary} and
Lemma \ref{lemma5}.
\end{proof}
\begin{example} [European call option]
Consider for $K>0$ the function $\Phi(x)=(x-K)^+$ which represents
the payoff of an European call option. Since $\Phi$ is convex, and
$-\Phi$ concave, we can deduce by means of the last corollary
\begin{align*}
h_{up}^0(x)=&~E^{P^{\overline{\sigma}}}\left((S_T^{0,x}-K)^+\gamma_T^{-1}\right),\\
h_{low}^0(x)=&~-E^{P^{\underline{\sigma}}}\left(-(S_T^{0,x}-K)^+\gamma_T^{-1}\right).
\end{align*}
Furthermore, the function
\begin{align*}
u(t,x):=E^{P^{\overline{\sigma}}}\left((S_T^{t,x}-K)^+\gamma_{T-t}^{-1}\right),
\quad (t,x)\in [0,T]\times \mathbb{R}_+,
\end{align*}
solves the PDE in \eqref{BSB1}. The function
\begin{align*}
\underline{u}(t,x):=E^{P^{\underline{\sigma}}}\left(-(S_T^{t,x}-K)^+\gamma_{T-t}^{-1}\right),
\quad (t,x)\in [0,T]\times \mathbb{R}_+,
\end{align*}
solves Equation \eqref{BSB1} with $\overline{\sigma}$ replaced by
$\underline{\sigma}$ and boundary condition
$\underline{u}(T,x)=-(x-K)^+ ~\forall x\in\mathbb{R}_+$.
\end{example}

If $\Phi$ exhibits mixed convexity/concavity behavior meaning
that, for instance, there exists an $x^\star\in \mathbb{R}_+$ such
that $\Phi\upharpoonright_{[0,x^\star]}$ is convex whereas
$\Phi\upharpoonright_{[x^\star, \infty]}$ is concave, the
situation is much more involved.\par For instance in the case when
$\Phi$ represents a bullish call spread as considered in
\cite{avellaneda} the worst--case volatility will switch between
the volatility bounds $\underline{\sigma}$ and $\overline{\sigma}$
at some threshold $\bar{x}(t)$. The $t$ indicates the time
dependence of the threshold. This fact can be verified by solving
the PDE in \eqref{BSB} numerically, see \cite{avellaneda}.\par
Clearly, the evaluation of $\Phi$ becomes economic relevant when
$\Phi$ represents complex derivatives or a whole portfolio which
combines long and short positions. Pricing the whole portfolio is
more efficient than pricing the single positions separately and
leads to more reasonable results for the no--arbitrage bounds
since the bounds are closer based on the subadditivity of $E_G$.
Numerical methods for solving the Black--Scholes--Barenblatt PDE
in \eqref{BSB} can be found in \cite{barenblatt}.\\

\section{Conclusion}\label{section6}
We present a general framework in mathematical finance in order to
deal with model risk caused by volatility uncertainty. This
encompasses the extension of terminology widely used in Finance
like portfolio strategy, consumption process, arbitrage prices and
the concept of no--arbitrage. It is being modified to a
quasi--sure analysis framework resulting from the presence of
volatility uncertainty.\par Our setting does not involve any
reference measure and hence does not exclude any economic
interesting model a priori. We consider a stock price modeled by a
geometric G--Brownian motion which features volatility uncertainty
based on the structure of a G--Brownian motion. In this ambiguous
financial setting we examine the pricing and hedging of European
contingent claims. The ``G--framework" summarized in \cite{peng1}
gives us a meaningful and appropriate mathematical setting. By
means of a slightly new concept of no--arbitrage we establish
detailed results which provide a better economic understanding of
financial markets under volatility uncertainty.
\par The current paper may form the basis
for examining economic relevant questions in the presence of
volatility uncertainty in the sense that it extends important
notions in Finance and shows how to control. Concrete examples are
problems like hedging under constraints (cf. \cite{kk1}) and
portfolio optimization (cf. \cite{merton}). A natural step is to
extend above results to American contingent claims and then, for
instance, consider entry decisions of a firm in the sense of
irreversible investments as in \cite{ozaki} who solved the problem
in the presence of drift uncertainty.\par By the natural
properties of sublinear expectation any sublinear expectation
induces a coherent risk measure, see \cite{peng1}. G--expectation
may appear as a natural candidate to measure model risk. In this
context one might also imagine many concrete applications in
Finance.


\appendix

\section{Sublinear expectations}\label{section2}
We depict notions and preliminaries in the theory of sublinear
expectation and related G--Brownian motion. This includes the
definition of G--expectation, introduction to It\^o calculus with
G--Brownian motion and important results concerning the
representation of G--expectation and G--martingales. We do not
express definitions and results in their most generality. Our task
rather is to present it in a manner of which we used it in the
previous sections. More details can be
found in \cite{peng1} and \cite{pengli}.\\
We also restrict ourselves to the one-dimensional case. However,
everything also holds in the $d$--dimensional case. Also the
financial market model can be extended to d risky assets using a
d--dimensional G--Brownian motion as it is done in classical
financial markets with Brownian motion.

\subsection{Sublinear expectation, G--Brownian motion and
G--expectation}\label{A1}

\begin{definition}
Let $\Omega\neq\emptyset$ be a given set. Let $\mathcal{H}$ be a
linear space of real valued functions defined on $\Omega$ with
$c\in\mathcal{H}$ for all constants $c$ and $|X|\in\mathcal{H}$ if
$X\in\mathcal{H}$. ($\mathcal{H}$ can be considered as the space
of random variables.) A sublinear expectation $\hat{E}$ on
$\mathcal{H}$ is a functional $\hat{E}:\mathcal{H}\rightarrow
\mathbb{R}$ satisfying the following properties: For any $X,
Y\in\mathcal{H}$ we have
\begin{itemize}
\item [(a)] Monotonicity: If $X\geq Y$ then
$\hat{E}(X)\geq\hat{E}(Y)$. \item [(b)] Constant preserving:
$\hat{E}(c) = c$.  \item [(c)] Sub-additivity: $\hat{E}(X+Y)\leq
\hat{E}(X)+\hat{E}(Y)$. \item [(d)] Positive homogeneity:
$\hat{E}(\lambda X)=\lambda\hat{E}(X)\quad\forall \lambda\geq 0$.
\end{itemize}
The  triple $(\Omega,\mathcal{H},\hat{E})$ is called a sublinear
expectation space.
\end{definition}

Property (c) is also called self--domination. It is equivalent to
$\hat{E}(X)-\hat{E}(Y)\leq \hat{E}(X-Y)$. Property (c) together
with (d) is called sublinearity. It implies convexity:
$$\hat{E}\left(\lambda X + (1 -\lambda)Y\right)\leq\hat{E}(X) + (1
- \lambda)E(Y) \text{ for any } \lambda\in [0,1].$$ The properties
(b) and (c) imply cash translatability: $$\hat{E}\left(X + c
\right) = \hat{E}(X) + c \text{ for any } c \in\mathbb{R}.$$ The
space $C_{l,Lip}(\mathbb{R}^n)$, where $n\geq 1$ is an integer,
plays an important role. It is the space of all real-valued
continuous functions $\varphi$ defined on $\mathbb{R}^n$ such that
$|\varphi(x)-\varphi(y)|\leq C(1+|x|^k+|y|^k)|x-y|~~\forall
x,y\in\mathbb{R}^n$. Here $k$ is an integer depending on
$\varphi$.

\begin{definition}
In a sublinear expectation space $(\Omega,\mathcal{H},\hat{E})$ a
random variable $Y \in\mathcal{H}$ is said to be independent from
another random variable $X\in \mathcal{H}$ under $\hat{E}$ if for
any test function $\varphi\in C_{l,Lip}(\mathbb{R}^2)$ we have
\begin{align*}
\hat{E}[\varphi(X,Y)]=\hat{E}[\hat{E}[\varphi(x,Y)]_{x=X}].
\end{align*}
\end{definition}

\begin{definition}
Let $X_1$ and $X_2$ be two random variables defined on sublinear
expectation spaces $(\Omega_1,\mathcal{H}_1,\hat{E}_1)$ and
$(\Omega_2,\mathcal{H}_2,\hat{E}_2)$, respectively. They are
called identically distributed, denoted by $X_1\sim X_2$, if
\begin{align*}
\hat{E}_1[\varphi(X_1)]=\hat{E}_2[\varphi(X_2)]\quad\forall
\varphi\in C_{l,Lip}(\mathbb{R}).
\end{align*}
We call $\bar{X}$ an independent copy of $X$ if $\bar{X}\sim X$
and $\bar{X}$ is independent from $X$.
\end{definition}
\begin{definition}[G--normal distribution]\label{hier}
A random variable $X$ on a sublinear expectation space
$(\Omega,\mathcal{H},\hat{E})$ is called (centralized) G--normal
distributed if for any $a,b\geq 0$
\begin{align*}
aX+b\bar{X}~\sim~\sqrt{a^2+b^2}X
\end{align*}
where $\bar{X}$ is an independent copy of $X$. The letter $G$
denotes the function
\begin{align*}
G(y):=\frac{1}{2}\hat{E}[yX^2]:\mathbb{R}\rightarrow\mathbb{R}.
\end{align*}
\end{definition}

Note that $X$ has no mean-uncertainty, i.e., one can show that
$\hat{E}(X)=\hat{E}(-X)=0$. Furthermore, the following important
identity holds
\begin{align*}
G(y)=\frac{1}{2}\overline{\sigma}^2y^+-\frac{1}{2}\underline{\sigma}^2y^-
\end{align*}
with $\underline{\sigma}^2:=-\hat{E}(-X^2)$ and
$\overline{\sigma}^2:=\hat{E}(X^2)$. We write $X$ is
$N(\{0\}\times [\underline{\sigma}^2,\overline{\sigma}^2])$
distributed. Therefore we sometimes say that G--normal
distribution is characterized by the parameters
$0<\underline{\sigma}\leq \overline{\sigma}$.

\begin{remark}\label{remark_sigma}
All along the paper we assume $\underline{\sigma}>0$. From an
economic point of view this assumption is quite reasonable. In
Finance, volatility is always assumed to be greater zero. A
volatility equal to zero would induce arbitrage.\par The
G--framework also works without this condition. But based on this
assumption we get along in our paper without the notion of
viscosity solution. Our assumption ensures that the function $G$
is non--degenerate and therefore all involved PDEs induced by the
G--normal distribution, cf. Equation \eqref{PDEG}, become
classical $C^{1,2}$--solutions, see page 19 in \cite{peng1}.
\end{remark}

\begin{remark}
The random variable $X$ defined in \ref{hier} is also
characterized by the following parabolic partial differential
equation (PDE for short) defined on
$[0,T]\times \mathbb{R}$: \\
For any $\varphi\in C_{l,Lip}(\mathbb{R})$ define
$u(t,x):=\hat{E}[\varphi(x+\sqrt{t}X)]$, then $u$ is the unique
(viscosity) solution of
\begin{align}\label{PDEG}
\partial_tu-G(\partial_{xx}u)=0,\quad u(0,\cdot)=\varphi(\cdot).
\end{align}
The PDE is called a G--equation.
\end{remark}

\begin{definition}
Let $(\Omega,\mathcal{H},\hat{E})$ be a sublinear expectation
space. $(X_t)_{t\geq 0}$ is called a stochastic process if $X_t$
is a random variable in $\mathcal{H}$ for each $t\geq 0$.
\end{definition}

\begin{definition}[G--Brownian motion]
A process $(B_t)_{t\geq 0}$ on a sublinear expectation space
$(\Omega,\mathcal{H},\hat{E})$ is called a G--Brownian motion if
the following properties are satisfied:
\begin{itemize}
\item [(i)] $B_0=0$. \item [(ii)] For each $t,s\geq 0$ the
increment $B_{t+s} -B_t$ is $N(\{0\}\times
[\underline{\sigma}^2s,\overline{\sigma}^2s])$ distributed and
independent from $(B_{t_1},B_{t_2},\cdots,B_{t_n})$ for each
$n\in\mathbb{N}$, $0\leq t_1\leq\cdots\leq t_n\leq t$.
\end{itemize}
\end{definition}

Condition (ii) can be replaced by the following three conditions
giving a characterization of G--Brownian motion: \begin{itemize}
\item [(i)] For each $t,s\geq 0$: $B_{t+s}-B_t\sim B_t$ and
$\hat{E}(|B_t|^3)\rightarrow 0$ as $t\rightarrow 0$. \item [(ii)]
The increment $B_{t+s}-B_t$ is independent from
$(B_{t_1},B_{t_2},\cdots,B_{t_n})$ for each $n\in\mathbb{N}$ and
$0\leq t_1\leq\cdots\leq t_n\leq t$. \item [(iii)]
$\hat{E}(B_t)=-\hat{E}(-B_t)=0\quad\forall t\geq 0$.
\end{itemize}\par
For each $t_0>0$ we have that $(B_{t+t_0}-B_{t_0})_{t\geq 0}$
again is a G--Brownian motion.\\
\par

Let us briefly depict the construction of G--expectation and its
corresponding G--Brownian motion. As in the previous sections we
fix a time horizon $T>0$ and set $\Omega_T=C_0([0,T],\mathbb{R})$
-- the space of all real--valued continuous paths starting at
zero.
We will consider the canonical process
$B_t(\omega):=\omega_t,t\leq T,\omega\in\Omega$. We define
\begin{align*}
L_{ip}(\Omega_T):=\{\varphi(B_{t_1},\cdots,B_{t_n})|n\in\mathbb{N},t_1,\cdots,t_n\in
[0,T],\varphi\in C_{l,Lip}(\mathbb{R}^n)\}.
\end{align*}
A G--Brownian motion is firstly constructed on $L_{ip}(\Omega_T)$.
For this purpose let $(\xi_i)_{i\in\mathbb{N}}$ be a sequence of
random variables on a sublinear expectation space
$(\tilde{\Omega},\tilde{\mathcal{H}},\tilde{E})$ such that
$\xi_{i}$ is G--normal distributed and $\xi_{i+1}$ is independent
of $(\xi_1, \cdots,\xi_i)$ for each integer $i\geq 1$.\par Then a
sublinear expectation on $L_{ip}(\Omega_T)$ is constructed by the
following procedure: For each $X\in L_{ip}(\Omega_T)$ with ${X =
\varphi(B_{t_1} - B_{t_0},B_{t_2} - B_{t_1},\cdots,B_{t_n} -
B_{t_{n-1}})}$ for some $\varphi\in C_{l,Lip}(\mathbb{R}^n)$, $0
\leq t_0 < t_1 < \cdots < t_n \leq T$, set {\small\begin{align*}
E_G [\varphi(B_{t_1}-B_{t_0},B_{t_2} - B_{t_1},\cdots,B_{t_n} -
B_{t_{n-1}})]
:=\tilde{E}[\varphi(\sqrt{t_1-t_0}\xi_1,\cdots,\sqrt{t_n-t_{n-1}}\xi_n)].
\end{align*}}
The related conditional expectation of $X \in L_{ip}(\Omega_T)$ as
above under $\Omega_{t_i},i\in\mathbb{N}$, is defined by
{\small\begin{align*} E_G [\varphi(B_{t_1}-B_{t_0},B_{t_2} -
B_{t_1},\cdots,B_{t_n} - B_{t_{n-1}})|\Omega_{t_i}]
:=\psi(B_{t_1}-B_{t_0},\cdots,B_{t_i} - B_{t_{i-1}})
\end{align*}}
where
$\psi(x_1,\cdots,x_i):=\tilde{E}[\varphi(x_1,\cdots,x_i,\sqrt{t_{i+1}-t_i}\xi_{i+1},\cdots,\sqrt{t_n-t_{n-1}}\xi_n)].$
One checks that $E_G$ consistently defines a sublinear expectation
on $L_{ip}(\Omega_T)$ and the canonical process $B$ represents a
G--Brownian motion.
\begin{definition}
{The sublinear expectation $E_G: L_{ip}(\Omega_T)\rightarrow
\mathbb{R}$} defined through the above procedure is called a
G-–expectation. The corresponding canonical process $(B_t)_{t\in
[0,T]}$ on the sublinear expectation space
$(\Omega_T,L_{ip}(\Omega_T),E_G)$ is a G-–Brownian motion.
\end{definition}

Let $||\xi||_p:=[E_G(|\xi|^p)]^{\frac{1}{p}}$ for $\xi\in
L_{ip}(\Omega_T),p\geq 1$. Then for any $t\in [0,T],
E_G(\cdot|\Omega_t)$ can be continuously extended to
$L_G^p(\Omega_T)$ -- the completion of $L_{ip}(\Omega_T)$ under
the norm $||\xi||_p$.
\begin{proposition}
The conditional G--expectation
$E_G(\cdot|\Omega_t):L^1_G(\Omega_T)\rightarrow L^1_G(\Omega_t)$
defined above has the following properties: For any $t\in
[0,T],X,Y\in L^1_G(\Omega_T)$ we have
\begin{itemize}
\item [(i)] $E_G(X|\Omega_t)\geq E_G(Y|\Omega_t)$ if $X\geq Y$.
\item [(ii)] $E_G(\eta|\Omega_t)=\eta$ if $\eta\in
L_G^1(\Omega_t)$. \item [(iii)]
$E_G(X|\Omega_t)-E_G(Y|\Omega_t)\leq E_G(X-Y|\Omega_t)$. \item
[(iv)] ${E_G(\eta X|\Omega_t)=\eta^+E_G(X|\Omega_t)+\eta^-
E_G(-X|\Omega_t) \text{ for each bounded }\eta\in
L_G^1(\Omega_t).}$ \item [(v)]
$E_G\left(E_G\left(X|\Omega_t\right)|\Omega_s\right)=E_G(X|\Omega_{t\wedge
s})$. \item [(vi)] $E_G(X|\Omega_t)=E_G(X)$ for each
$L_G^1(\Omega_T^t)$.
\end{itemize}
\end{proposition}

The following property is often very useful. Of course, it holds
for any sublinear expectation if the related conditional
expectation is defined reasonably.
\begin{proposition}\label{prop2}
Let $X,Y\in L_G^1(\Omega_T)$ with
$E_G(Y|\Omega_t)=-E_G(-Y|\Omega_t)$ for some $t\in [0,T]$. Then we
have
\begin{align*}
E_G(X+Y|\Omega_t)=E_G(X|\Omega_t)+E_G(Y|\Omega_t).
\end{align*}In particular, if $E_G(Y|\Omega_t)=E_G(-Y|\Omega_t)=0$
then we have \\ $E_G(X+Y|\Omega_t)=E_G(X|\Omega_t)$.
\end{proposition}

G--expectation and its corresponding G--Brownian motion is not
based on a given classical probability measure. The next theorem
establishes the ramification with probability theory. As a
consequence we obtain a set of probability measures which
represents the functional $E_G$ in a subsequently announced sense.
Although the measures belonging to the set are mutually singular
this result is similar to the classical ambiguity setting when the
probability measures inducing the ambiguity are absolutely
continuous, see \cite{ec}, \cite{delbaen}. References for the
representation theorem for G--expectation are \cite{pengdenis} and
\cite{penghu}.\par Let $\mathcal{F}=\mathcal{B}(\Omega_T)$ be the
Borel $\sigma$--algebra and consider the probability space
$(\Omega_T,\mathcal{F},P)$. Let $W=(W_t)$ be a classical Brownian
motion in this space. The filtration generated by $W$ is denoted
by $(\mathcal{F}_t)$ where ${\mathcal{F}_t:=\sigma\{W_s|0\leq
s\leq t\}\vee \mathcal{N}}$ and $\mathcal{N}$ denotes the
collection of $P$--null subsets. For fixed $t\geq 0$ we also
denote $\mathcal{F}_s^t:=\sigma\{W_{t+u}-W_t|0\leq u\leq s\}\vee
\mathcal{N}$.\par Let
$\Theta:=[\underline{\sigma},\overline{\sigma}]$ such that
$G(y)=\frac{1}{2}\sup_{\theta\in\Theta}y\theta^2$ and denote by
$\mathcal{A}_{t,T}^\Theta$ the collection of all $\Theta$--valued
$(\mathcal{F}_s^t)$--adapted processes on $[t,T]$. For any
$\theta\in\mathcal{A}_{t,T}^\Theta$ we define
\begin{align*}
B_T^{t,\theta}:=\int_t^T \theta_s dW_s.
\end{align*}
Let $P^\theta$ be the law of the process
$B_t^{0,\theta}=\int_0^t\theta_sdW_s,t\in [0,T]$, i.e.,
$P^\theta=P\circ (B^{0,\theta})^{-1}$. Define
$\mathcal{P}_1:=\{P^\theta|\theta\in\mathcal{A}_{0,T}^\Theta\}$
and (the weakly compact set)
$\mathcal{P}:=\overline{\mathcal{P}}_1$ as the closure of
$\mathcal{P}_1$ under the topology of weak convergence.
\par
Using these notations we can formulate the following result.
\begin{theorem}\label{theorem_representation}
For any $\varphi\in C_{l,Lip}(\mathbb{R}^n),n\in\mathbb{N}, 0\leq
t_1\leq \cdots\leq t_n\leq T,$ we have
\begin{align*}
E_G[\varphi(B_{t_1},\cdots,B_{t_n}-B_{t_{n-1}})]=&
\sup_{\theta\in\mathcal{A}_{0,T}^\Theta}E^P[\varphi(B_{t_1}^{0,\theta},\cdots,B_{t_n}^{t_{n-1},\theta})]\\
=&~\sup_{\theta\in\mathcal{A}_{0,T}^\Theta}E^{P^\theta}[\varphi(B_{t_1},\cdots,B_{t_n}-B_{t_{n-1}})]\\
=&~\sup_{P^\theta\in\mathcal{P}}E^{P^\theta}[\varphi(B_{t_1},\cdots,B_{t_n}-B_{t_{n-1}})].
\end{align*}
Furthermore,
\begin{align*}
E_G(X)=\sup_{P\in\mathcal{P}}E^P(X)\quad\forall X\in
L_G^1(\Omega_T).
\end{align*}
\end{theorem}

The last theorem can also be extended to the conditional
G--expectation, see also \cite{soner1}. For $X\in L^1_G(\Omega_T),
t\in [0,T]$, and $Q\in\mathcal{P}$,
\begin{align*}
E_G(X|\mathcal{F}_t)=
\esssup_{Q'\in\mathcal{P}(t,Q)}E^{Q'}(X|\mathcal{F}_t)\quad Q-a.s.
\end{align*}
where $\mathcal{P}(t,Q):=\{Q'\in\mathcal{P}|Q'=Q \text{ on
}\mathcal{F}_t\}$.\\
\par
As seen in the previous sections the
following terminology is very useful within the framework of
G--expectation.
\begin{definition}
A set $A\in\mathcal{F}$ is \emph{polar} if $P(A)=0$ for all
$P\in\mathcal{P}$. We say a property holds ``\emph{quasi-surely}"
(q.s.) if it holds outside a polar set.
\end{definition}

\cite{peng1} also gives a pathwise description of the space
$L^p_G(\Omega_T)$. This is quite helpful to get a better
understanding of the space. Before passing to the description we
need the following definition.
\begin{definition}
A mapping $X:\Omega_T\rightarrow\mathbb{R}$ is said to be
\emph{quasi--continuous} (q.c.) if $\forall \epsilon
> 0$  there exists an open set $O$ with $\sup_{P\in\mathcal{P}}P(O) < \epsilon$ such
that $X|_{O^c}$ is continuous. \par We say that $X:\Omega_T
\rightarrow \mathbb{R}$ has a \emph{quasi--continuous version} if
there exists a quasi--continuous function
$Y:\Omega_T\rightarrow\mathbb{R}$ with $X = Y$ q.s.
\end{definition}

\cite{peng1} showed that $L^p_G(\Omega_T),p>0,$ is equal to the
closure of the continuous and bounded functions on $\Omega_T$,
$C_b(\Omega_T)$, with respect to the norm
$||X||_p:=(\sup_{P\in\mathcal{P}}E^P[|X|^p])^{\frac{1}{p}}$.
Furthermore, the space $L^p_G(\Omega_T), p>0,$ is characterized by
{\small\begin{align}\label{pathwise} L^p_G(\Omega_T)=\{X\in
L^0(\Omega_T):X\text{ has a q.c. version, }
\lim_{n\rightarrow\infty}\sup_{P\in\mathcal{P}}E^P[|X|^p1_{\{|X|>n\}}]=0\}
\end{align}}
where $L^0(\Omega_T)$ denotes the space of all measurable
real--valued functions on $\Omega_T$.\\
\par The mathematical framework provided enables the analysis of
stochastic processes for several mutually singular probability
measures simultaneously. Therefore, when not stated otherwise all
equations are also to be understood in the sense of ``quasi-sure".
This means that a ``property" holds almost-surely for all
conceivable scenarios.

\subsection{Stochastic calculus of It\^o type with G--Brownian
motion}\label{A2}

We briefly present the basic notions on stochastic calculus like
the construction of It\^o's integral with respect to G--Brownian
motion.\\
For $p\geq 1$, let $M^{p,0}_G(0,T)$ be the collection of all
simple processes $\eta$ of the following form: Let
$\{t_0,t_1,\cdots,t_N\},N\in\mathbb{N}$, be a partition of
$[0,T]$, $\xi_i\in L^p_G(\Omega_{t_i})~ \forall i=0,1,\cdots,N-1$.
Then for any $t\in [0,T]$ the process $\eta$ is defined by
\begin{align}\label{form}
\eta_t(\omega):=\sum_{j=0}^{N-1}\xi_j(\omega)
1_{[t_j,t_{j+1})}(t).
\end{align}
For each $\eta\in M^{p,0}_G(0,T)$ let
$||\eta||_{M_G^p}:=\left(E_G\int_0^T
|\eta_s|^pds\right)^{\frac{1}{p}}$ and denote by $M_G^p(0,T)$ the
completion of $M^{p,0}_G(0,T)$ under the norm $||\cdot||_{M^p_G}$.
\begin{definition}
For $\eta\in M^{2,0}_G(0,T)$ with the presentation in \eqref{form}
we define the integral mapping $I:M^{2,0}_G(0,T)\rightarrow
L^2_G(\Omega_T)$ by
\begin{align*}
I(\eta)=\int_0^T\eta(s)dB_s:=\sum_{j=0}^{N-1}
\xi_j(B_{t_{j+1}}-B_{t_j}).
\end{align*}
\end{definition}

Since $I$ is continuous it can be continuously extended to
$M^2_G(0,T)$. The integral has similar properties as in the
classical It\^o calculus case. For more details see
\cite{peng1}.\par The quadratic variation process of $B$ is
defined like in the classical case as the limit of the quadratic
increments. The following identity holds
\begin{align*}
\langle B\rangle_t=B_t^2-2\int_0^tB_sdB_s\quad\forall t\leq T.
\end{align*}
The quadratic variation of $B$ is a continuous, increasing process
which is absolutely continuous with respect to $dt$. $(\langle
B\rangle_t)$ contains all the statistical uncertainty of $B$. It
is a typical process with mean uncertainty. For $s,t\geq 0$ we
have $\langle B\rangle_{s+t}-\langle B\rangle_s\sim \langle
B\rangle_t$ and it is independent of $\Omega_s$. Furthermore, for
any $t\geq s\geq 0$
\begin{align*}
E_G[\langle B\rangle_t-\langle
B\rangle_s|\Omega_s]=&~\overline{\sigma}^2(t-s),\\E_G[-(\langle
B\rangle_t-\langle
B\rangle_s)|\Omega_s]=&-\underline{\sigma}^2(t-s).
\end{align*}
We say that $\langle B\rangle_t$ is
$N([\underline{\sigma}^2t,\overline{\sigma}^2t]\times
\{0\})$--distributed, i.e., for all $\varphi\in
C_{l,Lip}(\mathbb{R})$,
\begin{align*}
E_G[\varphi(\langle B\rangle_t)]=\sup_{\underline{\sigma}^2\leq
v\leq \overline{\sigma}^2} \varphi(vt).
\end{align*}
The integral with respect to the quadratic variation of
G--Brownian motion $\int_0^t \eta_sd\langle B\rangle_s$ is defined
in an obvious way. Firstly, for all $\eta\in M^{1,0}_G(0,T)$ and
again by a continuity argument for all $\eta\in M^1_G(0,T)$.\\
\par
The following observation is important for the characterization of
G--martingales. The It\^o integral can also be defined for the
following processes, see \cite{song1}: Let $H_G^0(0,T)$ be the
collection of processes $\eta$ having the following form: For a
partition $\{t_0,t_1,\cdots,t_N\}$ of $[0,T], N\in\mathbb{N}$, and
$\xi_i\in L_{ip}(\Omega_{t_i})~ \forall i=0,1,\cdots,N-1$, let
$\eta$ be given by
\begin{align*}
\eta_t(\omega):=\sum_{j=0}^{N-1}\xi_j(\omega)
1_{[t_j,t_{j+1})}(t)\quad\forall t\leq T.
\end{align*}
For $p\geq 1$ and $\eta\in H_G^0(0,T)$ let $||\eta||_{H^p_G}:=
\left(E_G\left(\int_0^T|\eta_s|^2ds\right)^{\frac{p}{2}}\right)^{\frac{1}{p}}$
and denote by $H_G^p(0,T)$ the completion of $H_G^0(0,T)$ under
this norm $||\cdot||_{H_G^p}$. In the case $p=2$ the spaces
$H^2_G(0,T)$ and $M^2_G(0,T)$ coincide. As before we can construct
It\^o's integral $I$ on $H_G^0(0,T)$ and extend it to $H_G^p(0,T)$
for any $p\geq 1$ continuously, hence $I: H_G^p(0,T)\rightarrow
L^p_G(\Omega_T)$.

\subsection{Characterization of G--martingales}\label{A3}

\begin{definition}\label{G--martingale}
A process $M=(M_t)$ with values in $L^1_G(\Omega_T)$ is called
G--martingale if $E_G(M_t|\mathcal{F}_s)=M_s$ for all $s,t$ with
$s\leq t\leq T$. If $M$ and $-M$ are both G--martingales $M$ is
called a symmetric G--martingale.
\end{definition}

By means of the characterization of the conditional G--expectation
we have that $M$ is a G--martingale if and only if for all $0\leq
s\leq t\leq T, P\in\mathcal{P}$,
\begin{align*}
M_s=\esssup_{Q'\in\mathcal{P}(s,P)}E^{Q'}(M_t|\mathcal{F}_s)\quad
P-a.s.
\end{align*}
cf. \cite{soner1}. This identity declares that a G--martingale $M$
can be seen as a multiple prior martingale which is a
supermartingale for any $P\in\mathcal{P}$ and a martingale for an
optimal measure.
\par
The next results give a characterization for G--martingales.
\begin{theorem}
Let $x\in\mathbb{R},z\in M^2_G(0,T)$ and $\eta\in M_G^1(0,T)$.
Then the process \begin{align*} M_t:= x+\int_0^tz_sdB_s+\int_0^t
\eta_sd\langle B\rangle_s-\int_0^t 2G(\eta_s)ds, \quad t\leq T,
\end{align*}
is a G--martingale.
\end{theorem}

In particular, the nonsymmetric part
${-K_t:=\int_0^t\eta_sd\langle B\rangle_s-\int_0^t2G(\eta_s)ds,}$
$t\in [0,T],$ is a G--martingale which is quite surprising
compared to classical probability theory since $(-K_t)$ is
continuous, non--increasing with quadratic variation equal to
zero.
\begin{remark}\label{K=0}
$M$ is a symmetric G--martingale if and only if $K\equiv 0$, see
also \cite{song1}.
\end{remark}
\begin{theorem}[Martingale representation]\label{martingale} (\cite{song1})
Let $\beta\geq 1$ and $\xi\in L_G^\beta (\Omega_T)$. Then the
G--martingale $X$ with $X_t:=E_G(\xi|\mathcal{F}_t),t\in [0,T]$,
has the following unique representation
\begin{align*}
X_t=X_0+\int_0^tz_sdB_s-K_t
\end{align*}
where $K$ is a continuous, increasing process with $K_0=0, K_T\in
L_G^\alpha(\Omega_T), z\in H_G^\alpha(0,T), \forall \alpha\in
[1,\beta)$, and $-K$ a G--martingale.
\end{theorem}

If $\beta=2$ and $\xi$ bounded from above we get that $z\in
M^2_G(0,T)$ and $K_T\in L^2_G(\Omega_T)$, see \cite{song2}.


\bibliographystyle{econometrica}
\ifx\undefined\BySame
\newcommand{\BySame}{\leavevmode\rule[.5ex]{3em}{.5pt}\ }
\fi \ifx\undefined\textsc
\newcommand{\textsc}[1]{{\sc #1}}
\newcommand{\emph}[1]{{\em #1\/}}
\let\tmpsmall\small
\renewcommand{\small}{\tmpsmall\sc}
\fi


\end{document}